\documentclass[pre,amsmath,amssymb,twocolumn,showpacs,superscriptaddress]{revtex4-2}

\usepackage{graphicx}
\usepackage{dcolumn}
\usepackage{bm}
\usepackage{scalerel}
\usepackage{amsthm}
\usepackage{hyperref}
\usepackage{xcolor}
\usepackage[capitalise]{cleveref}
\usepackage[mathlines]{lineno}

\newtheorem{theorem}{Theorem}

\newtheorem{lemma}{Lemma}[section]

\newcommand\wh[1]{\widehat{#1}}
\newcommand{\wt}{\widetilde}

\renewcommand{\P}{\mathbf{P}}
\newcommand{\E}{\mathbf{E}}
\DeclareMathOperator{\1}{\mathbf{1}}
\newcommand{\eps}{\epsilon}

\DeclareMathOperator{\Var}{\mathbf{Var}}
\DeclareMathOperator{\Cov}{\mathbf{Cov}}
\DeclareMathOperator{\Corr}{\mathbf{Corr}}

\DeclareMathSymbol{\shortminus}{\mathbin}{AMSa}{"39}

\begin{document}


\title{Correlation thresholds in the steady states of particle systems and spin glasses
}%

\author{Jacob Calvert}
\email{calvert@gatech.edu}
\affiliation{Institute for Data Engineering and Science, Georgia Institute of Technology, Atlanta, GA, USA}
\affiliation{Santa Fe Institute, Santa Fe, NM, USA}
\author{Dana Randall}
\email{randall@cc.gatech.edu}
\affiliation{School of Computer Science, Georgia Institute of Technology, Atlanta, GA, USA}
\affiliation{Santa Fe Institute, Santa Fe, NM, USA}

\date{\today}

\begin{abstract}
A growing body of theoretical and empirical evidence shows that the global steady-state distributions of many equilibrium and nonequilibrium systems approximately satisfy an analogue of the Boltzmann distribution, with a local dynamical property of states playing the role of energy. The correlation between the effective potential of the steady-state distribution and the logarithm of the exit rates determines the quality of this approximation. We demonstrate and explain this phenomenon in a simple one-dimensional particle system and in random dynamics of the Sherrington--Kirkpatrick spin glass by providing the first explicit estimates of this correlation. We find that, as parameters of the dynamics vary, each system exhibits a threshold above and below which the correlation dramatically differs. We explain how these thresholds arise from underlying transitions in the relationship between the local and global ``parts'' of the effective potential.
\end{abstract} 

\maketitle


\section{Introduction}\label{sec:intro}

The steady-state distributions of many physical systems favor a relatively small number of special states. For example, this is true of proteins, which tend to adopt conformations of low energy \cite{Dill1990}, as well as active matter systems, like swarms of interacting robots that sustain long-lasting, near-periodic ``dances'' \cite{Chvykov2021}. For proteins, and for systems in thermal equilibrium in general, the Boltzmann distribution explains that states with lower energy are exponentially favored. More precisely, the steady-state probability $\pi (x)$ of a state $x$ satisfies
\begin{equation}\label{eq:boltzmann}
\pi (x) = e^{-\beta H(x)} / Z,
\end{equation}
in terms of inverse temperature $\beta$, energy or Hamiltonian~$H$, and partition function $Z = \sum_x e^{-\beta H(x)}$.

The explanation of order that the Boltzmann distribution provides is powerful because the energy $H(x)$ is a ``local'' property of state $x$, in the sense that it can be determined without observing long trajectories of the system to states far from $x$. For example, the energy of a protein conformation can be estimated from its coordinates using molecular dynamics force fields \cite{noe_boltzmann_2019}. In contrast, the steady-state distribution of a nonequilibrium system is ``global'' because estimating it from a trajectory generally requires the observation of multiple returns to state $x$, which can entail visits to distant states (e.g., see \cite{lee_computing_2013}). For nonequilibrium steady states, there can be no local explanation of order analogous to the one the Boltzmann distribution provides, as Landauer explained \cite{landauer_inadequacy_1975}. Essentially, the issue with nonequilibrium steady states is that changes to the dynamics in one part of the state space can affect the relative probability of states in another \cite{landauer_motion_1988}.

Despite this barrier, a local dynamical property of states called ``rattling'' predicts the steady-state distributions of many nonequilibrium systems \cite{Chvykov2018,Chvykov2021,gold_self-organized_2021,yang_emergent_2022,kedia_drive-specific_2023}. Informally, rattling is a measure of how quickly the system exits the vicinity of a state, like a diffusion coefficient in abstract state space, but for processes that are not necessarily diffusive. (For an introduction to rattling, see the supplementary materials of \cite{Chvykov2021}.) In their study of robot swarms, Chvykov et al.\ observed that plotting the effective potential $-\log \pi (x)$ against the rattling $R(x)$ of discretized swarm configurations $x$ produced a roughly linear scatter \cite{Chvykov2021}. This observation is analogous to an approximate Boltzmann distribution because the effective potential and $H(x)$ are exactly collinear in equilibrium steady states (\cref{fig1}).

\begin{figure}
    \centering
    \includegraphics[width=1\linewidth]{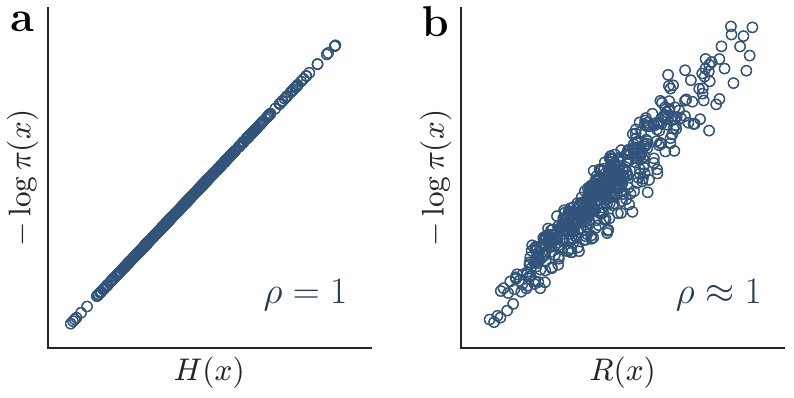}
    \caption{Local--global correlations. (\textbf{a}) The Boltzmann distribution entails perfect correlation $\rho$ between the effective potential of an equilibrium steady state and the energy of a uniformly random state. (Each mark corresponds to one state.) (\textbf{b}) For many equilibrium and nonequilibrium steady states, a different local property of a state, called rattling, is highly correlated with the effective potential.
    }
    \label{fig1}
\end{figure}

These observations motivate the study of correlations between the effective potential of a steady state and local properties of the states. In recent work \cite{Calvert2024}, we derived a general formula for the kind of correlation that Chvykov et al.\ observed, which \cref{fig1}b depicts. It applies to Markov jump processes, which are commonly used to model the dynamics underlying equilibrium and nonequilibrium steady states \cite{schnakenberg_network_1976,zia_probability_2007,Barato2015,owen_universal_2020,RevModPhys.97.015002}. The formula can be used to derive lower bounds on the typical correlation exhibited by broad classes of steady states, such as those of reaction kinetics on disordered energy landscapes \cite{calvert2025localglobalcorrelationsdynamicsdisordered}. To complement these results, which are somewhat abstract, this paper provides the first concrete calculations of the local--global correlation. Although we focus on two specific models, we expect that our approach will apply to many other models as well.

\section{Local--global correlation in Markov jump processes}\label{sec:pnas}

\subsection{Defining local--global correlation}\label{sub:defining_local_global_correlation}

We consider a Markov jump processes on a finite set~$S$ of discrete states, defined by a transition rate matrix $(W_{x,y})_{x,y \in S}$. The probability $p_y (t)$ of finding the system in state $y$ after $t$ units of time solves the master equation
\begin{equation}
    \dot{p}_y (t) = \sum_{x \in S} W_{x,y} \, p_x (t).
\end{equation}
We assume that the transition rates strongly connect the state space, which guarantees the existence and uniqueness of the steady-state or stationary distribution $\pi = \lim_{t \to \infty} p(t)$ \cite{norris_markov_1997}. Our focus is on the relationship between the effective potential $- \log \pi (x)$ and the logarithm of the exit rates, defined by 
\begin{equation}
q(x) = \sum_{y: \, y \neq x} W_{x,y}.
\end{equation}
The quantity $\log q(x)$ is the analogue of rattling $R(x)$ for Markov jump processes \cite{Calvert2024}. Accordingly, to understand the collinearity of scatter plots like \cref{fig1}b, we analyze the correlation between the effective potential and log exit rate of a uniformly random state $X \in S$:
\begin{equation}\label{eq:corr_def}
    \rho = \Corr (-\log \pi (X), \log q(X)).
\end{equation}
We assume that $\pi$ and $q$ are non-constant to ensure that~$\rho$ is well-defined. For reference, the linear correlation of random variables $U$ and $V$ with positive, finite variances is defined as the normalized covariance
\begin{equation}\label{eq:book_corr_def}
    \Corr(U,V) = \frac{\Cov (U,V)}{\sqrt{\Var (U) \Var (V)}}.
\end{equation}
The covariance equals $\Cov (U,V) = \E(UV) - \E(U) \E(V)$, in terms of the expected value $\E (U)$ of $U$ and variance $\Var (U) = \Cov (U,U)$.

\subsection{The effective potential's local and global parts}\label{subsec:explaining_rattling_correlation}

Associated with the preceding continuous-time Markov jump process is a discrete-time Markov jump process, defined by normalizing the transition rates out of each state:
\begin{equation}\label{eq:jump chain def}
    \wh W_{x,y} = W_{x,y}/q(x).
\end{equation}
This process also has a unique stationary distribution $\wh\pi$, which is related to $\pi$ and $q$ by
\begin{equation}\label{eq:three_quantites}
    \pi (x) = \frac{\wh\pi (x) / q (x)}{\sum_{y \in S} \wh\pi (y)/q(y)}.
\end{equation}
Intuitively, the limiting fraction $\pi (x)$ of time spent in state $x$ is proportional to the limiting fraction $\wh\pi (x)$ of visits made to $x$, multiplied by the typical duration of a visit to $x$, which equals $1/q (x)$. Taking the logarithm of both sides of \cref{eq:three_quantites} shows that 
\begin{equation}\label{eq:local_global_parts}
-\log \pi (X) = -\log \wh\pi (X) + \log q (X) + \mathrm{constant}.
\end{equation}
Following the discussion in \cref{sec:intro}, \cref{eq:local_global_parts} justifies the view that $-\log \wh\pi (X)$ and $\log q(X)$ are the global and local ``parts'' of the effective potential.

The authors in~\cite{Calvert2024} used \cref{eq:local_global_parts} to derive the formula 
\begin{equation}\label{eq:corr}
\rho = \frac{1+\wh\rho r}{\sqrt{1+2\wh\rho r + r^2}},
\end{equation}
in terms of
\begin{equation}\label{eq:rhohat def}
  \wh\rho = \Corr (-\log \wh\pi (X), \log q (X)),
\end{equation}
the correlation between the global and local parts, and 
\begin{equation}\label{eq:r def}
  \quad r = \sqrt{\frac{\Var (\log \wh\pi (X))}{\Var (\log q (X))}},
\end{equation}
the ratio of their standard deviations. Note that $\wh\rho \in [-1,1]$ and $r \in [0,\infty)$.

\begin{figure}
    \centering
    \includegraphics[width=0.9\linewidth]{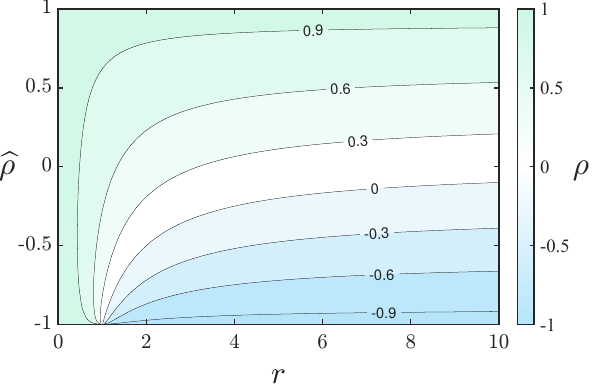}
    \caption{Contour plot of \cref{eq:corr}. Contours and shading indicate the value of $\rho$ for a given pair $(r,\wh\rho\,) \in [0,\infty) \times [-1,1]$. 
    }
    \label{fig2}
\end{figure}

Note that \cref{eq:corr} does not necessarily make $\rho$ easier to calculate---indeed, it should be as difficult to calculate $\wh\rho$ and~$r$ as it is to calculate $\rho$ from its definition \eqref{eq:corr_def}. \cref{eq:corr} does, however, imply a useful lower bound of $\rho$ in terms of~$r$ alone, which is tractable for broad classes of dynamics \cite{calvert2025localglobalcorrelationsdynamicsdisordered}. More importantly, \cref{eq:corr} explains how high correlation arises, in terms of the relationship between the parts of the effective potential.

\cref{fig2}, which depicts \cref{eq:corr}, shows that there are two ways for the correlation~$\rho$ to be close to~$1$ (the dark green region beyond the $\rho = 0.9$ contour). First, if the global part of the effective potential varies far less than the local part (i.e., $r \ll 1$), then $\rho \approx 1$, regardless of the value of $\wh\rho$. Second, if the parts of the effective potential are highly correlated (i.e., $\wh\rho \approx 1$), then $\rho \approx 1$, regardless of the value of $r$. For~$\rho$ to be negative, the parts of the effective potential must be sufficiently anti-correlated: $\wh\rho < -1/r$.

\section{Particles on a ring}\label{sec: asep}

We first demonstrate the calculation of $\rho$ for a toy model before analyzing more complicated spin-glass dynamics. The model consists of two particles moving clockwise around a discrete ring at rates depending on the gaps between the particles. We show that the local--global correlation is either close to~$1$ or close to~$-1$ depending on a parameter that biases the particles' motion. As we explain below, this model is equivalent to the Glauber dynamics of a single particle moving in a quadratic potential, at a temperature that depends on the bias parameter.

Consider two identical particles that move clockwise around a ring with $L+2$ discrete sites, where $L \geq 1$ is an integer. Each particle jumps clockwise to the next open site at a rate that depends on the number of open sites ``ahead'' of it. Specifically, if there are $y \geq 1$ open sites in the clockwise direction from one particle to the next, then the first particle jumps at a rate of $\alpha^y$, where $\alpha > 0$ is a parameter of the dynamics. If there are no open sites ahead of it, then the particle does not move. We arbitrarily choose one particle to be the ``first'' particle, and use~$x$ to denote the number of open sites to the second. Note that $x$, which takes values in $S = \{0,\dots,L\}$, completely determines the state of the system.

Based on the preceding description, the transition rate $W_{x,y}$ from $x$ to $y$ open sites satisfies
\begin{equation}\label{eq:asep_rates}
    W_{x,y} = 
    \begin{cases}
        \alpha^x & y = x-1,\\
        \alpha^{L-x} & y = x+1,\\ 
        0 & y \notin \{x-1,x+1\},
    \end{cases}
\end{equation}
for $x,y \in S$. For example, the number of open sites decreases from $x$ to $x-1$ when the first particle jumps, which happens at a rate of $W_{x,x-1} = \alpha^x$. The corresponding exit rates are
\begin{equation}\label{eq: asep q def}
    q (x) = \sum_{y: \, y \neq x} W_{x,y} = 
    \begin{cases}
        \alpha^L & x \in \{0,L\},\\ 
        \alpha^x + \alpha^{L-x} & x \notin \{0,L\}.
    \end{cases}
\end{equation}
As before, we denote the stationary distribution of the Markov chain with these transitions rates by $\pi$. Note that, although $W_{x,y}$, $q$, and $\pi$ depend on the parameter~$\alpha$, we omit it from the notation.

\cref{eq:asep_rates} equivalently describes the Glauber dynamics of a particle in a harmonic trap. Specifically, the non-zero transition rates in \cref{eq:asep_rates} can be expressed as 
\begin{equation}\label{eq: asep glauber rates}
W_{x,x\pm 1} = \nu \exp\left(- \frac{\beta}{2} (E (x\pm 1) - E(x))\right),
\end{equation}
in terms of the potential $E(x) = (x-L/2)^2$, inverse temperature $\beta = \log \alpha$, and constant $\nu = e^{\beta (L+1)/2}$. As a consequence, the stationary distribution $\pi$ is the Boltzmann distribution, which satisfies
\begin{equation}\label{eq:alpha exp}
\frac{\pi(x)}{\pi (0)} = e^{-\beta (E (x) - E(0))} = e^{-\beta (x^2 - Lx)} = \alpha^{x(L-x)}.
\end{equation}

The main result of this section is an estimate of the local--global correlation that the particle system exhibits.

\begin{theorem}\label{thm:main result asep}
For every $\alpha \neq 1$, the correlation satisfies 
\begin{equation}\label{eq:main result asep}
    \rho (\alpha) = \frac{\sqrt{15}}{4} \, \mathrm{sign}(\alpha-1) \left(1+O(L^{-1})\right),
\end{equation}
as $L \to \infty$.
\end{theorem}

\cref{thm:main result asep} says that $\rho (\alpha)$ is close to $\sqrt{15}/4 \approx 0.97$ when $\alpha > 1$ and close to $-\sqrt{15}/4$ when $\alpha < 1$, for all sufficiently large $L$ in terms of $\alpha$. (In \cref{eq:main result asep}, we use the standard asymptotic notation $O(L^{-1})$ to denote a quantity that is at most a constant multiple of $L^{-1}$ in absolute value for all sufficiently large $L$.) \cref{fig3}a shows a representative scatter plot of the effective potential against the log exit rates when $\alpha > 1$. For such values of $\alpha$, the correlation is already close to its limiting value when $L = 100$. For $\alpha < 1$, the correlation approaches $-\sqrt{15}/4$ more slowly (\cref{fig3}b).

\begin{figure}
    \centering
    \includegraphics[width=1\linewidth]{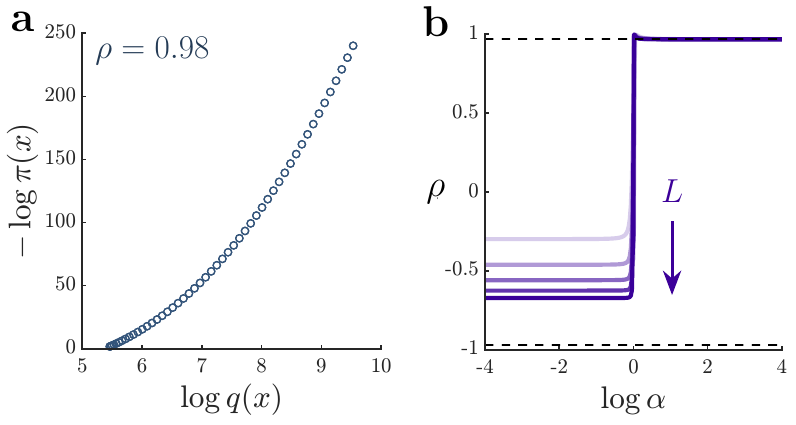}
    \caption{Local--global correlations of particles on a ring. (\textbf{a}) The effective potential scattered against the log exit rate of each state for the dynamics in \cref{eq:asep_rates} with $L = 100$ and $\alpha = 1.1$. (\textbf{b}) The correlation as a function of $L$ and $\alpha$ for $L = 25k$, $k \in \{2,\dots,6\}$. Dashed lines indicate $\rho = \pm \sqrt{15}/4$.
    }
    \label{fig3}
\end{figure}

The local--global correlation compares the exponential dependence of $\pi(x)$ and $q(x)$ on $x$. The reason why $\rho (\alpha)$ does not approach $\mathrm{sign}(\alpha-1)$ as $L \to \infty$ is because the former is quadratic, while the latter is roughly piecewise linear. Indeed, the error term in \cref{eq:main result asep} arises in part from approximating the exit rate $q(x)$ in \cref{eq: asep q def} by $\alpha^{\max\{x,L-x\}}$ when $\alpha > 1$, and $\alpha^{\min\{x,L-x\}}$ when $\alpha < 1$. The proof of \cref{thm:main result asep} in \cref{app:asep} finds that this approximation is sufficiently accurate when $L$ is large relative to the quantity $1/\sqrt{|\alpha^4 - 1|}$. Hence, if $\alpha \geq 2$ or $\alpha \leq 1/2$, for example, then $L$ only needs to be larger than some constant.

\subsection{Calculating the correlation}\label{subsec:asep_calc}

We outline the calculation of $\rho$ in this subsection, deferring some details to \cref{app:asep}.  Recalling \cref{eq:corr_def}, we need expressions for the stationary distribution and exit rates that have simple logarithms. This is true of the stationary distribution because $\pi (x)$ is proportional to $\alpha^{x(L-x)}$, so we start by simplifying the exit rates.

When $L$ and the extent of the bias $|\log \alpha\,|$ are large, the exit rates in \cref{eq: asep q def} are roughly equal to 
\begin{equation}\label{eq: tilde q asep}
\wt q(x) = 
\begin{cases}
    \alpha^{\max\{x,L-x\}} & \alpha > 1,\\ 
    \alpha^{\min\{x,L-x\}} & \alpha < 1,
\end{cases}
\end{equation}
which has a simpler logarithm than $q(x)$. \cref{fig4} compares~$q$ and~$\wt q$. In particular, it highlights the fact that $\wt q(x)$ underestimates $q(x)$ for some values of $x$ and $\alpha$, and overestimates it for others. However, aside from when $\alpha < 1$ and $x \in \{0,L\}$ (indicated by red arrows in \cref{fig4}b), $\wt q(x)$ approximates $q(x)$ increasingly well as $L$ increases.

\begin{figure}
    \centering
    \includegraphics[width=1\linewidth]{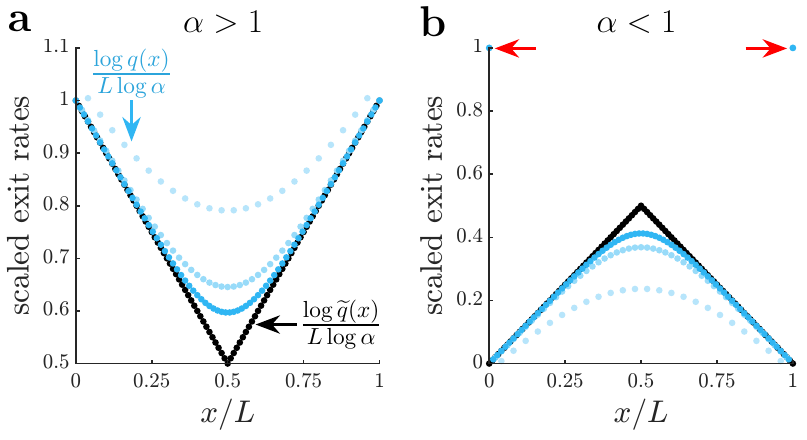}
    \caption{Comparison of $\log q (x)$ (blue points) and $\log \widetilde{q} (x)$ (black points), scaled by $1/L \log \alpha$, for (\textbf{a}) $\alpha = 1.1$ and (\textbf{b}) $\alpha = 0.9$. The blue scatters correspond to $L \in \{25,50,75\}$ (darker blue indicates higher $L$); the black scatter has $L = 100$. Note the points in (\textbf{b}) indicated by red arrows, which all blue scatters share.
    }
    \label{fig4}
\end{figure}

For the moment, we treat this approximation as exact and consider the correlation \eqref{eq:corr_def} with $\wt q$ in the place of $q$: 
\begin{equation}\label{eq: asep rho tilde1}
\wt \rho = \Corr (-\log \pi (X), \log \wt q(X)).
\end{equation}
(Recall that $X$ denotes a uniformly random state in $\{0,\dots, L\}$, as opposed to an arbitrary deterministic state~$x$.) We consider the $\alpha > 1$ case first.

According to \cref{eq:alpha exp,eq: tilde q asep}, the stationary distribution and exit rates satisfy
\begin{equation}\label{eq: pi and q}
\pi (X) = \pi (0) \, \alpha^{-U} \quad \text{and} \quad \wt q(X) = \alpha^{\wt V},
\end{equation}
in terms of the random variables $U = X(X-L)$ and $\wt V = \max\{X,L-X\}$.
Substituting these expressions into \cref{eq: asep rho tilde1}, we find that
\begin{equation}\label{eq: asep rho tilde2}
\wt \rho = \Corr (U \log \alpha - \log \pi (0), \, \wt V \log \alpha).
\end{equation}
Since the quantity ${-}\log \pi (0)$ and factors of $\log \alpha$ are constants, they do not affect the correlation, which therefore satisfies
\begin{equation}\label{eq:simple_rho_asep}
\wt \rho = \Corr (U,\wt V) = \frac{\Cov (U,\wt V)}{\sqrt{\Var (U) \Var (\wt V)}}.
\end{equation}

To calculate $\wt \rho$, we use the following estimates:
\begin{align} 
    \Cov (U,\wt V) &= \frac{1}{96} L^3 + O(L^2), \label{eq: asep est1}\\
    \Var (U) &= \frac{1}{180}L^4 + O(L^3), \label{eq: asep est2}\\
        \Var (\wt V) &= \frac{1}{48}L^2 + O(L). \label{eq: asep est3}
\end{align}
We obtain \cref{eq: asep est1,eq: asep est2,eq: asep est3} by direct calculation using the fact that $X$ is uniformly random over the $L+1$ states of $S$. For example, we calculate the covariance 
\begin{equation}
\Cov (U, \wt V) = \E (U \wt V) - \E (U) \E (\wt V)
\end{equation}
using the estimates
\begin{align}
\E (U \wt V) &= \frac{1}{L+1} \sum_{x=0}^L x(x-L) \max\{x,L-x\} \approx -\frac{11}{96} L^3, \nonumber\\
\E(U) &= \frac{1}{L+1} \sum_{x=0}^L x(x-L) \approx -\frac{1}{6} L^2,\nonumber\\
\E(\wt V) &= \frac{1}{L+1} \sum_{x=0}^L \max\{x,L-x\} \approx \frac{3}{4} L,
\end{align}
where ``$\approx$'' denotes equality to leading order in $L$. The covariance consequently satisfies
\begin{equation}
\Cov (U, \wt V) \approx -\frac{11}{96} L^3 + \frac{1}{6} L^2 \cdot \frac{3}{4} L = \frac{1}{96}
L^3.
\end{equation}
Similar calculations yield the variance estimates.

Substituting \cref{eq: asep est1,eq: asep est2,eq: asep est3} into \cref{eq:simple_rho_asep} shows that, when $\alpha > 1$, 
\begin{equation}
\wt\rho = \frac{(L^3/96) \left(1+O(L^{-1})\right)}{\sqrt{(L^4/180) (L^2/48)}} = \frac{\sqrt{15}}{4} \left(1+O(L^{-1})\right).
\end{equation}
In the $\alpha < 1$ case, following \cref{eq: tilde q asep}, we instead set $\wt V = \min\{X,L-X\}$. Using the same argument as in the $\alpha > 1$ case, we find that $\wt \rho = -(\sqrt{15}/4) (1+O(L^{-1}))$ because $\Cov(U,\wt V)$ changes sign, while $\Var (\wt V)$ is unchanged.

As we prove in \cref{app:asep}, $\wt\rho$ closely approximates $\rho$ when $L$ is sufficiently large as a function of $\alpha$. The error of this approximation, combined with the preceding estimates of $\wt\rho$, establishes \cref{thm:main result asep}.

\subsection{A threshold for the correlation}\label{subsec:key_quantities_asep}

Recall that \cref{eq:corr} explains how the correlation $\rho$ arises in terms of $\wh\rho$, the correlation between the global and local parts of the effective potential, and $r$, the ratio of their standard deviations. In particular, \cref{eq:corr} implies that~$\rho$ approximately equals $\wh \rho$ when $r \gg 1$ (i.e., the global part varies far more than the local part). As we now explain, for any $\alpha \neq 1$, the dynamics in \cref{eq:asep_rates} has $r \gg 1$ for all sufficiently large $L$. The threshold in $\rho (\alpha)$ that \cref{thm:main result asep} identifies therefore arises from an underlying threshold in the the correlation between the parts of the effective potential.

By definition, the quantities $\wh\rho$ and $r$ depend on the stationary distribution $\wh\pi$ of the discrete-time Markov jump process with transition rate matrix $\wh{W}$ \eqref{eq:jump chain def}. Specifically, they feature the quantity $\log \wh \pi (X)$, where $X$ is a uniformly random state. To calculate it, we recall \cref{eq:three_quantites}, which states that $\wh\pi (X)$ is proportional to $\pi (X) q(X)$. We then approximate $q$ by $\wt q$ and apply the expressions for $\wt q(X)$ and $\pi (X)$ from \cref{eq: pi and q} to find that
\begin{align}
    \log q(X) &\approx \wt V \log \alpha, \label{eq: logq approx}\\ 
    \log \wh \pi (X) &\approx (\wt V - U) \log \alpha + C, \label{eq: loghatpi approx}
\end{align}
for a constant $C$. As in the preceding subsection, $U$ denotes $X(X-L)$, while $\wt V$ denotes $\max\{X,L-X\}$ when $\alpha > 1$ and $\min \{X,L-X\}$ when $\alpha < 1$.

We use \cref{eq: logq approx,eq: loghatpi approx} to calculate $r$ as
\begin{equation}
r = \sqrt{\frac{\Var (\log \wh\pi (X))}{\Var (\log q(X))}} \approx \sqrt{\frac{\Var (\wt V - U)}{\Var (\wt V)}}.
\end{equation}
The quantity $\Var (\wt V - U)$ is roughly $\Var (U)$ because the variance of $U$ scales faster with $L$ than that of $\wt V$. To be precise, the estimates from \cref{eq: asep est2,eq: asep est3} imply that
\begin{align}
\sqrt{\frac{\Var (\wt V - U)}{\Var (\wt V)}} &= \sqrt{\frac{L^4/180}{L^2/48}} \left(1+O(L^{-1}) \right) \nonumber\\ 
&= \frac{2L}{\sqrt{15}} \left(1+O(L^{-1}) \right).
\end{align}
As a consequence, for any fixed $\alpha \neq 1$, $r \gg 1$ for all sufficiently large $L$. \cref{eq:corr} then implies that $\rho$ is primarily determined by the correlation $\wh\rho$ between the parts of the effective potential.

To establish that $\wh\rho$ exhibits a threshold in $\alpha$, we substitute \cref{eq: logq approx,eq: loghatpi approx} into the definition \eqref{eq:rhohat def} of $\wh\rho$, finding
\begin{align}
\wh\rho &= \Corr (-\log \wh\pi(X), \log q(X)) \nonumber\\ 
& \approx \Corr ( (U - \wt V) \log \alpha - C, \, \wt V \log \alpha).
\end{align}
The constant $C$ and factors of $\log \alpha$ do not affect the correlation, hence
\begin{equation}
    \wt\rho \approx \Corr (U - \wt V, \wt V).
\end{equation}
As the variance of $U$ grows faster with~$L$ than that of $\wt V$ does, $\Corr(U - \wt V, \wt V)$ is roughly $\Corr(U, \wt V)$ when $L$ is large. The sign of this correlation depends on whether $\wt V$ is a maximum or a minimum, hence whether $\alpha$ is above or below $1$. More precisely, the proof of \cref{thm:main result asep} establishes that
\begin{equation}
\Corr \big(U - \wt V, \wt V\big) \approx \frac{\sqrt{15}}{4} \mathrm{sign}(\alpha-1)\big(1+O(L^{-1})\big),
\end{equation}
for all $\alpha \neq 1$. We conclude that the threshold in $\rho$ arises because the local and global parts of the effective potential themselves undergo a transition in correlation as $\alpha$ varies.

\section{Spin-glass dynamics}

For our second example, we analyze a significantly more complicated family of dynamics associated with the Sherrington--Kirkpatrick (SK) spin glass. These dynamics are precisely the ``complex, nonlinear, and high-dimensional'' kind that Chvykov et al.\ suggested would be the domain of rattling \cite{Chvykov2021}. Our second main result shows that something subtler is true. Whether the non-equilibrium analogue of the Boltzmann distribution in \cref{fig1} applies depends sensitively on a parameter of the dynamics. Moreover, the threshold in correlation arises similarly to that of the particle system in \cref{sec: asep}.

The dynamics will be random, so the correlation $\rho$ will be random as well. Like \cref{thm:main result asep}, the main result of this section identifies a threshold in $\rho$, as a parameter of the dynamics varies. Although the threshold in $\rho$ is similar to that of the particle system, our subsequent analysis of $\wh\rho$ and $r$ shows that the underlying relationship between the local and global parts of the effective potential is entirely different.

For an integer $N \geq 1$ and inverse temperature $\beta > 0$, the SK model is the random probability distribution $\pi$ on the $N$-dimensional hypercube $S = \{-1,1\}^N$ that assigns to a configuration $x \in S$ the Boltzmann probability $\pi (x) \propto \exp(-\beta H(x))$, in terms of a random Hamiltonian~$H$. The Hamiltonian of $x$ is defined as 
\begin{equation}
H (x) = \sum_{i,j \in [N]} g_{ij} x_i x_j,
\end{equation}
in terms of random couplings $(g_{ij})_{i, j \in [N]}$ of the coordinates, which are independent and identically distributed (i.i.d.) normal random variables with mean $0$ and variance $1/N$.

We consider a family of Glauber dynamics for the SK model, like those commonly used in Monte Carlo simulations thereof \cite{mathieu_convergence_2000}. The dynamics are defined in terms of a parameter $\lambda \in [0,1]$ as
\begin{equation}\label{eq:sk model dynamics}
W_{x,y}^\lambda = e^{\beta\left( \lambda H(x) - (1-\lambda)H(y) \right)}, \quad y \sim x.
\end{equation}
Note that $y \sim x$ indicates that $y$ neighbors $x$ in the hypercube, meaning they differ in exactly one coordinate. This family interpolates between two extremes, in a sense. The first extreme entails transitions from $x$ to~$y$ at a rate of $W_{x,y}^{0} = e^{-\beta H(y)}$. In other words, jumps occur at a rate that depends on the energy of the destination. The second extreme entails rates that depend on the originating state: $W_{x,y}^{1} = e^{\beta H(x)}$.

We are again interested in the relationship between the exit rates
\begin{equation}
q^\lambda(x) = \sum_{y:\, y \sim x} W_{x,y}^\lambda,
\end{equation}
and the stationary distribution, which is the Boltzmann distribution $\pi$ for every $\lambda \in [0,1]$. To verify that $\pi$ is the stationary distribution, note that it satisfies detailed balance for a generic $\lambda$:
\begin{align}
\pi (x) W_{x,y}^\lambda &= e^{-\beta H(x) + \beta\left( \lambda H(x) - (1-\lambda)H(y) \right)} / Z \nonumber\\ 
&= e^{-\beta H(y) + \beta\left( \lambda H(y) - (1-\lambda)H(x) \right)} / Z \nonumber\\ 
&= \pi (y) W_{y,x}^\lambda.
\end{align}

Since the dynamics depend on the random Hamiltonian $H$, the correlation defined by
\begin{equation}
    \rho (\lambda) = \Corr (-\log \pi (X), \log q^\lambda (X)),
\end{equation}
is a random variable. Here, $X$ denotes a uniformly random state of $S$ that is independent of the random couplings. The main result of this section is an estimate of $\rho$ that holds at high temperatures, with a probability that tends to $1$ as the dimension $N$ grows.

\begin{theorem}\label{thm:sk}
Let $N \geq 3$ be an integer and set $\lambda_\ast(N) = \frac{1-4/N}{2-4/N}$. As $\beta \to 0$, the correlation satisfies 
\begin{equation}\label{eq:sk}
\rho (\lambda) = \mathrm{sign}(\lambda - \lambda_\ast) \left(1 + \frac{\sqrt{N}}{|\lambda-\lambda_\ast|} \cdot O(\beta) \right)
\end{equation}
for every $\lambda \neq \lambda_\ast$, with a probability of at least $1 - O(N^{-2})$ as $N \to \infty$.
\end{theorem}

Informally, \cref{thm:sk} means that, in high enough dimensions and at high enough temperatures, the correlation is approximately $\mathrm{sign}(\lambda - \lambda_\ast)$ for every $\lambda \neq \lambda_\ast$, with a probability close to $1$ over the randomness of the couplings. According to \cref{eq:sk}, the correlation abruptly increases from $-1$ to $1$ as $\lambda$ increases past $\lambda_\ast$ (\cref{fig5}). Note that $\lambda_\ast$ lies in $[0,1]$ only if $N \geq 5$. In particular, \cref{eq:sk} implies that the correlation is nonnegative regardless of the value of $\lambda$ when the dimension $N$ is at most $4$ (\cref{fig5}c). As in \cref{thm:main result asep}, our estimate of the correlation features an error term that arises from an approximation of the exit rates by a function that is easier to analyze.

In the low-temperature regime, it is more difficult to obtain a useful approximation of the exit rates. However, by combining the methods of this paper with those of \cite{calvert2025localglobalcorrelationsdynamicsdisordered}, we expect that it is possible to obtain a lower bound of~$\rho$ that tends to $1$ as $\lambda$ does, which holds with high probability as $N \to \infty$. In particular, we anticipate that, in both the high- and low-temperature regimes, $\rho$ is typically close to~$1$ when~$\lambda$ is.

\begin{figure}
    \centering
    \includegraphics[width=1\linewidth]{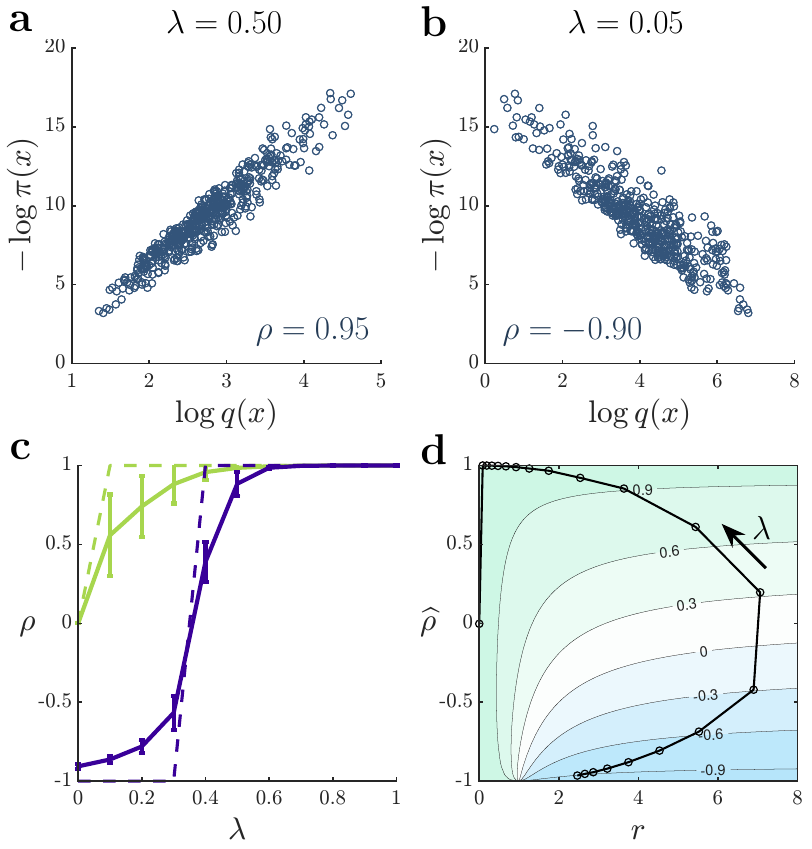}
    \caption{Local--global correlations of SK model dynamics. The effective potential scattered against the log exit rate of each state for (\textbf{a}) $\lambda = 0.50$ and (\textbf{b}) $\lambda = 0.05$. (\textbf{c}) The correlation as a function of $\lambda$ for $N = 4$ (green curve) and $N=10$ (blue curve). The curves show the mean correlation over $50$ and $10$ independent trials, respectively, with corresponding error bars of $\pm 1$ standard deviation. Dashed lines linearly interpolate the formula in \cref{eq:sk}. (\textbf{d}) Plot of $(r,\wh{\rho})$ pairs for $\lambda \in [0,1]$ with $N=10$ fixed. Each mark represents the average over $10$ independent trials of $\wh\rho$ and $r$ for a particular~$\lambda$. Shading and contours indicate the corresponding value of~$\rho$. The arrow indicates the direction of increasing $\lambda$. In all cases, $\beta = 1$.
    }
    \label{fig5}
\end{figure}

\subsection{Calculating the correlation at high temperature}\label{subsec:sk_calc}

For the moment, assume that the random Hamiltonian~$H$ is given, in which case the only randomness is in the uniformly random state $X$. We begin by expressing the exit rates as
\begin{equation}
q^\lambda(x) = \sum_{k \in [N]} e^{\beta\left( \lambda H(x) - (1-\lambda)H(x^{(k)}) \right)},
\end{equation}
in terms of $x^{(k)}$, the element of $S$ that agrees with $x$ except in the $k$th coordinate:
\begin{equation}
x^{(k)} = (x_1, \dots, x_{k-1}, -x_k, x_{k+1}, \dots, x_N).
\end{equation}
We then separate a factor of $e^{\beta (2\lambda - 1) H(x)}$ from the rates, which results in 
\begin{equation}\label{eq:qlambda}
q^\lambda(x) = e^{\beta (2\lambda - 1) H(x) + \beta A(x)},
\end{equation}
in terms of the quantity
\begin{equation}
A(x) = \frac{1}{\beta} \log \sum_{k \in [N]} e^{-\beta (1-\lambda) (H(x^{(k)}) - H(x))}.
\end{equation}
By substituting \cref{eq:qlambda} and the Boltzmann distribution $\pi (x) \propto e^{-\beta H(x)}$ into \cref{eq:corr_def}, we find that the correlation satisfies
\begin{equation}\label{eq:corr_in_sk}
    \rho = \Corr \big( H(X), (2\lambda-1) H(X) + A(X)\big).
\end{equation}
The correlation is therefore determined by the variances of $H(X)$ and $A(X)$, as well as their covariance.

While the variance of $H(X)$ can be calculated exactly, the quantities involving $A$ are more complicated. We will resort to estimates of these quantities in the high-temperature setting, based on the following idea. A short calculation shows that, since the couplings are i.i.d.\ normal random variables with mean $0$ and variance $1/N$, the $N2^N$ random energy differences of the form $H(x^{(k)}) - H(x)$ that define $A$ will rarely exceed $\sqrt{N}$ in absolute value. Hence, if $\beta$ is sufficiently small in terms of $N$, then $A(x)$ will be close to 
\begin{equation}
\wt A (x) = \frac{1}{\beta} \log N - \frac{1-\lambda}{N} \sum_{k \in [N]} (H(x^{(k)}) - H(x)),
\end{equation}
uniformly in $x$. We will show that the closeness of $A$ and~$\wt A$ implies that $\rho$ is close to
\begin{equation}
\wt \rho = \Corr \left( H(X), (2\lambda - 1) H(X) + \wt A (X) \right).
\end{equation}
Accordingly, instead of calculating $\rho$ directly, we will calculate $\wt\rho$ and then bound the error incurred by doing so.

We implement this strategy in \cref{app:sk}, where we further show that the variance of $\wt A$ satisfies
\begin{equation}\label{eq: sk a var}
\Var (\wt A(X)) = (4(1-\lambda)/N)^2 \Var (H(X)),
\end{equation}
while its covariance with $H$ satisfies
\begin{equation}\label{eq: sk cov ah}
\Cov (H(X),\wt A(X)) = (4(1-\lambda)/N) \Var (H(X)).
\end{equation}
Using these estimates, we calculate $\wt\rho$ when $\lambda \neq \lambda_\ast$ to be 
\begin{widetext}
    \begin{align}
    \wt\rho (\lambda) &= \frac{(2\lambda-1) \Var H(X) + \Cov (H(X),A(X))}{\sqrt{\Var H(X) ( (2\lambda-1)^2 \Var H(X) + \Var A(X) + 2 (2\lambda-1) \Cov (H(X),A(X)))}} \nonumber\\ 
    &= \frac{(2\lambda-1) + 4(1-\lambda)/N}{\sqrt{(2\lambda-1)^2 + 8(2\lambda-1)(1-\lambda)/N + (4(1-\lambda)/N)^2}} = \mathrm{sign} \big( (2\lambda-1) + 4(1-\lambda)/N \big) = \mathrm{sign} (\lambda - \lambda_\ast).\label{eq:sk rho tilde}
    \end{align}
\end{widetext}
This explains the formula in \cref{thm:sk}, aside from the error term, the control of which addresses the randomness of the couplings. We provide a full proof in \cref{app:sk}.

\subsection{Interpreting the correlation threshold}\label{subsec:key_quantities_sk}

As in \cref{subsec:key_quantities_asep}, we turn to the relationship between the parts of the effective potential. By \cref{eq:three_quantites}, $\wh\pi^{\lambda} (x)$ is proportional to $\pi (x) q^\lambda(x)$, hence
\begin{equation}
\wh\pi^\lambda (x) \propto e^{-\beta(2 (1-\lambda) H(x) - A(x))}. 
\end{equation}
We therefore find that
\begin{multline}
\wh\rho = \Corr \big(2(1-\lambda) H(X) - A(X),\\ (2\lambda-1) H(X) + A(X)\big),
\end{multline}
while the parameter $r$ satisfies
\begin{equation}
r = \sqrt{\frac{\Var (2(1-\lambda) H(X) - A(X))}{\Var( (2\lambda - 1) H(X) + A(X))}}.
\end{equation}

We could calculate $\wh\rho$ and $r$ directly, like we did in \cref{subsec:key_quantities_asep}. Instead, we note that \cref{eq:three_quantites} implies an analogue of \cref{eq:corr} for $\wh\rho$:
\begin{equation}
\wh\rho = - \frac{1-\rho s}{\sqrt{1-2\rho s + s^2}},
\end{equation}
where the corresponding analogue of $r$ is
\begin{equation}
s = \sqrt{\frac{\Var (\log \pi (X))}{\Var (\log q(X))}}.
\end{equation}
Note that $\wh\rho \approx \rho$ when $s \gg 1$. Approximating $A$ by~$\wt A$ and then using \cref{eq: sk a var,eq: sk cov ah}, we find that 
\begin{align}
s &= \sqrt{\frac{\Var H(X)}{\Var \big( (2\lambda-1) H(X) + A(X)\big)}} \nonumber\\ 
&\approx \frac{1}{\big|(2\lambda-1) + 4(1-\lambda)/N \big|} = \frac{1}{|\lambda - \lambda_\ast|}.
\end{align}
Heuristically, when $\lambda$ is sufficiently close to the threshold at $\lambda_\ast$, $s$ is much larger than $1$, in which case $\wh\rho$ is close to $\rho$. As in \cref{thm:main result asep}, we conclude that the transition from negative to positive $\rho$ is mediated primarily through a similar change in $\wh\rho$ (\cref{fig5}d).

\section{Discussion}\label{sec:disc}

Steady states that exhibit high local--global correlation satisfy an approximate analogue of the Boltzmann distribution (\cref{fig1}). In this case, a simple local property of a state---the logarithm of the exit rate---plays the role of energy. The virtue of this approximation is that exit rates are often easy to estimate. Indeed, the exit rate of a state is the reciprocal of the average holding time in it. To estimate the exit rate, it therefore suffices to repeatedly initialize the system in the state and average the times it takes to leave. This means that it is unnecessary to observe long trajectories of the system to faraway states, which is generally required of estimating the steady-state distribution \cite{lee_computing_2013}.

    While a major goal of local--global correlation is to explain order in broad classes of nonequilibrium steady states, the concept applies equally well to equilibrium steady states. For this reason, and to facilitate our calculations of $\rho$, we analyzed reversible dynamics with stationary distributions that are explicitly known. For some equilibrium steady states, the Hamiltonian and the logarithm of exit rates coincide. For example, when $\lambda = 1$, the exit rates of the SK model dynamics \eqref{eq:sk model dynamics} satisfy $q (x) = e^{\beta H(x)}$ and therefore $\rho = 1$. More generally, high local--global correlation can be valuable for understanding equilibrium steady states because exit rates may be readily estimable even when the Hamiltonian is complicated or unknown.

    Clearly, many steady states cannot exhibit high local--global correlation. It is therefore important to understand precisely why and how typically $\rho$ is close to $1$. The general formula for $\rho$ in \cref{eq:corr} somewhat abstractly explains that there are two ``ways'' for $\rho$ to be close to~$1$ (\cref{subsec:explaining_rattling_correlation}). These ways are defined in terms of key parameters $\wh\rho$ and $r$, which together characterize the relationship between the local and global parts of the effective potential (\cref{fig2}). Then, to understand how typically~$\rho$ is close to~$1$, one could estimate the typical values of $\wh\rho$ and $r$ under a probability distribution on a class of steady states. However, these quantities are generally difficult to calculate, and our recent work instead resorts to bounding their expected values  \cite{calvert2025localglobalcorrelationsdynamicsdisordered}. 

    In this context, \cref{thm:main result asep,thm:sk} serve two purposes. First, they provide explicit estimates of the local--global correlation in two families of dynamics, one deterministic and one random. The proofs of these estimates demonstrate an approach to calculating $\rho$ that we expect will apply to many other dynamics. Second, they show that the property of a steady state exhibiting high local--global correlation can depend sensitively on parameters of the dynamics (\cref{fig3}b and \cref{fig5}c). In both cases, the thresholds in $\rho$ arise from underlying thresholds in the correlation $\wh\rho$ between the parts of the effective potential. Such thresholds can also arise when $r$, which measures the relative variance of the parts, crosses~$1$ (e.g., \cite[Fig.\ 6]{Calvert2024}).

    An important goal for future work is to develop methods for estimating or bounding $\rho$ in situations where the stationary distribution $\pi$ is not explicit, as is generally true of dynamics that violate detailed balance. One example of recent work in this direction approximates the stationary distributions of nonreversible Markov chains with random transition rates by the probability distribution that is inversely proportional to the exit rates \cite{calvert2025noteasymptoticuniformitymarkov}. A possible strategy is to use the Markov chain tree theorem \cite{leighton_estimating_1986} to express and manipulate $\pi$, which has been integral to several recent results in stochastic thermodynamics \cite{owen_universal_2020,DalCengio2023,Floyd2025}. 
    
\begin{acknowledgments}
The authors thank two anonymous reviewers and Sid Redner for valuable input. This work was supported by U.S.\ Army Research Office award MURI W911NF-19-1-0233 and by the National Science Foundation award CCF-2106687.
\end{acknowledgments}

\section*{Data Availability}

The data that support the findings of this article are openly available \cite{calvert_2025_17791154}.

\appendix

\section{Approximation of correlation}\label{app:corr_perturb_est}

The correlations $\rho$ in \cref{thm:main result asep,thm:sk} are difficult to calculate exactly, due to the form of each model's exit rates. Instead, we calculate the correlations $\wt\rho$ resulting from a close approximation of the exit rates, which is simpler to analyze. The next lemma states the intuitive fact that, so long as this approximation is sufficiently accurate, $\wt \rho$ and $\rho$ are close. In the statement and its proof, we write $1 \pm a$ to denote the interval $[1-a,1+a]$ and $b = 1\pm a$ to mean that $b$ belongs to this interval. 

\begin{lemma}\label{lem:corr_perturb}
Let $U$, $V$, and $\wt V$ be random variables with finite, positive variances, that satisfy
\begin{equation}
| \Cov (U, \wt V) | \geq c \sqrt{\Var (U)} \quad \text{and} \quad \Var (\wt V) \geq c^2,
\end{equation}
for some $c > 0$. If $\Var (V - \wt V) < c^2/4$, then $\Corr (U,V)$ equals
\begin{equation}
    \Corr (U, \wt V) \left( 1 \pm \frac{5}{c} \sqrt{\Var (V - \wt V)} \right). 
\end{equation}
\end{lemma}

\begin{proof}
Let $\eps = V - \wt V$ and $a = (\Var (\eps)/\Var(\wt V))^{1/2}$. The Cauchy--Schwarz inequality implies that 
\begin{align}
\Cov (U,V) &= \Cov (U,\wt V) + \Cov (U,\eps) \nonumber\\ 
&= \Cov (U, \wt V) (1 \pm a).
\end{align}
The variance of $V$ analogously equals
\begin{align}
\Var (V) &= \Var(\wt V) + \Var(\eps) + 2 \Cov(\wt V, \eps) \nonumber\\ 
& = \Var(\wt V) \left( 1 + a^2 \pm 2 a \right).
\end{align}
Substituting these expressions into \cref{eq:book_corr_def} shows that
\begin{equation}
\Corr (U,V) = \Corr (U,\wt V) \left( \frac{1 \pm a}{1 + a^2 \pm 2a} \right).
\end{equation}
The assumed bounds on $\Var(\wt V)$ and $\Var (\eps)$ imply that $a < 1/4$. As a consequence, the parenthetical expression is always positive, and some algebra further shows that the interval it represents is contained in $[1 - 3a, 1 + 5a]$. This completes the proof, since $a \leq \sqrt{\Var(\eps)}/c$ by the assumed lower bound of $\Var(\wt V)$.
\end{proof}

\section{Proof of \texorpdfstring{\cref{thm:main result asep}}{Theorem 1}}\label{app:asep}

The proof outline in \cref{subsec:asep_calc} was incomplete because, instead of calculating $\rho$, we calculated
\begin{equation}
\wt \rho = \Corr (-\log \pi (X), \log \wt q(X)),
\end{equation}
where
\begin{equation}
\wt q(x) = 
\begin{cases}
    \alpha^{\max\{x,L-x\}} & \alpha > 1,\\ 
    \alpha^{\min\{x,L-x\}} & \alpha < 1.
\end{cases}
\end{equation}
We now use \cref{lem:corr_perturb} to obtain the estimate of $\rho$ in \cref{thm:main result asep} from our estimate of $\wt\rho$.

\begin{proof}[Proof of \cref{thm:main result asep}]
Recall that $\rho = \Corr (U,V)$, where 
\begin{equation}\label{eq:uv_def}
U = \frac{-\log \pi (X)}{\log \alpha} \quad \text{and} \quad V = \frac{\log q(X)}{\log \alpha}.
\end{equation}
We analogously define $\wt\rho = \Corr (U,\wt V)$, in terms of 
\begin{equation}
\wt V = \frac{\log \wt q (X)}{\log \alpha}, 
\end{equation}
and set $\eps = \wt V - V$. 

\cref{eq:simple_rho_asep} and the covariance estimates that immediately follow it show that
\begin{equation}
\wt \rho = \Corr(U,\wt V) = \frac{\sqrt{15}}{4} \, \mathrm{sign}(\alpha-1) \left(1+O(L^{-1})\right).
\end{equation}
The same estimates show that $U$ and $\wt V$ satisfy the hypotheses of \cref{lem:corr_perturb} with $c = \delta L$, for a small positive number $\delta$. To apply the conclusion of \cref{lem:corr_perturb}, it therefore suffices to show that $\eps$ has variance of less than $\delta^2 L^2/4$. In fact, $\Var (\eps)$ is much smaller as a function of~$L$. We consider the cases of $\alpha > 1$ and $\alpha < 1$ separately. 

First, let $\alpha > 1$. By definition, $\eps (x) = 0$ for $x \in \{0,L\}$, while 
\begin{equation}
\eps (x) = \log \left(1+\alpha^{\min\{x,L-x\}-\max\{x,L-x\}}\right)
\end{equation}
for $x \notin \{0,L\}$. Since $\log (1+x) \leq x$ for all $x > -1$, the second moment of $\eps$ satisfies 
\begin{equation}
\E(\eps^2) \leq \frac{1}{L+1} \sum_{x=1}^{L-1} \alpha^{2 (\min\{x,L-x\}-\max\{x,L-x\})}.
\end{equation}
Simplifying the geometric sum, and using the fact that $\Var(\eps) \leq \E(\eps^2)$, we find that
\begin{equation}
\Var(\eps) \leq \frac{2\alpha^4}{L(\alpha^4 - 1)}.
\end{equation}
This bound shows that, as long as $L$ is sufficiently large as a function of $\alpha$, then $\Var (\eps) < \delta^2 L^2/4$. \cref{lem:corr_perturb} then implies that 
\begin{equation}
\rho = \wt \rho \left(1 \pm \frac{5 \sqrt{2}}{\delta L^{3/2}} \sqrt{\frac{\alpha^4}{\alpha^4 - 1}} \right).
\end{equation}

When $\alpha < 1$, the only difference is that the minimum and maximum that appear in the preceding expression for $\eps$ are swapped. Analogous reasoning then shows that $\Var(\eps) \leq 2/(L(1-\alpha^4))$. Hence, $\Var (\eps) < \delta^2L^2/4$ for all sufficiently large $L$ as a function of $\alpha$, and so \cref{lem:corr_perturb} implies that the correlations are related by
\begin{equation}
\rho = \wt\rho \left(1 \pm \frac{5\sqrt{2}}{\delta L^{3/2}} \frac{1}{\sqrt{1 - \alpha^4}} \right).
\end{equation}

In both the $\alpha > 1$ and $\alpha < 1$ cases, by taking $L$ larger in terms of $\alpha$ if necessary, we can ensure that
\begin{equation}
\rho = \wt\rho \left(1 + O(L^{-1}) \right).
\end{equation}
Together with the estimate of $\wt\rho$, this proves that
\begin{equation}
\rho = \frac{\sqrt{15}}{4} \, \mathrm{sign}(\alpha-1) \left(1+O(L^{-1})\right).
\end{equation}
\end{proof}

\section{Proof of \texorpdfstring{\cref{thm:sk}}{Theorem 2}}\label{app:sk}

To complete the proof of \cref{thm:sk} outlined in \cref{subsec:sk_calc}, we need to bound the difference of $\rho$ and $\wt\rho$, and justify the expressions for the variance of $\wt A(X)$ and its covariance with $H(X)$, which we used to calculate $\wt\rho$. For the former task, we aim to apply \cref{lem:corr_perturb} with $\eps = A(X) - \wt A(X)$, which requires a bound on $\Var (\eps)$. We emphasize that $\Var(\eps)$ is a random variable, because the variance is taken with respect to the uniformly random state $X$, not the randomness of the couplings $(g_{ij})_{i,j \in [N]}$ that define the energy.

We will bound above $\Var(\eps)$ on the event that energy differences of the form $D_k (x) = H(x^{(k)}) - H(x)$ are at most $O(\sqrt{N})$ in absolute value, across all states $x \in S$ and coordinates $k \in [N]$. Recall that 
\begin{align}
A(x) &= \frac{1}{\beta} \log \sum_{k \in [N]} e^{-\beta (1-\lambda) D_k (x)},\\ 
\wt A (x) &= \frac{1}{\beta} \log N - \frac{1-\lambda}{N} \sum_{k \in [N]} D_k (x).
\end{align}
We will assume that $\beta$ is small enough as a function of~$N$ so that, when this even occurs, the summands that appear in $A(x)$ approximately satisfy
\begin{equation}\label{eq:summand_bound}
e^{-\beta (1-\lambda) D_k (x)} \approx 1 - \beta(1-\lambda) D_k (x),
\end{equation}
uniformly for $x \in S$. Approximating the summands of $A(x)$ in this way results in $\wt A(x)$, hence $\eps$ is the error associated with \cref{eq:summand_bound}.

\subsection{Control of energy differences}

The next lemma uses the fact that the $D_k (x)$ are centered normal random variables to prove that all energy differences are $O(\sqrt{N})$ in absolute value, with probability $1$ as $N \to \infty$. Throughout this section, as in the statement of \cref{thm:sk}, we assume that $N \geq 2$ is an integer.

\begin{lemma}\label{lem:e occurs}
    The largest absolute energy difference satisfies
    \begin{equation}
    \max_{x \in S, \, k \in [N]} \left| D_k (x) \right| \leq 4\sqrt{N},
    \end{equation}
    with a probability of at least $1 - N^{-2}$. 
\end{lemma}

\begin{proof}
By definition, $H(x) = \sum_{i, j \in [N]} g_{ij} x_i x_j$. The energy difference $D_k (x)$ therefore satisfies
\begin{align}
D_k (x) &= \sum_{i , j \in [N]} g_{ij} x_i^{(k)} x_j^{(k)} - \sum_{l, m \in [N]} g_{lm} x_l x_m \nonumber\\ 
&= \sum_{j \in [N] \setminus\{k\}} -2 (g_{jk} + g_{kj}) x_j x_k.\label{eq:diff_formula}
\end{align}
The couplings $g_{jk}$ are i.i.d.\ $\mathcal{N}(0,1/N)$ random variables, so the preceding expression shows that $D_k (x)$ is a sum of $N-1$ i.i.d.\ $\mathcal{N}(0,4/N)$ random variables. Hence, if $\sigma^2 = 4(1-\frac{1}{N})$, then $D_k (x)/\sigma$ is a standard normal random variable, which satisfies the tail bound
\begin{equation}\label{eq:d tail bd}
\P \left( |D_k (x)/\sigma| > t \right) \leq 2 e^{-t^2/2}, \quad t \geq 0.
\end{equation}

Anticipating a sum of the tail bound over $N2^N$ possible choices of state $x$ and coordinate $k$, we set
\begin{equation}
t = \sqrt{(2\log 2)N + 8 \log N},
\end{equation}
and define $\mathcal{E}_N (x,k)$ to be the event that $D_k (x)/\sigma$ exceeds~$t$ in absolute value:
\begin{equation}
\mathcal{E}_N (x,k) = \left\{ |D_k (x)/\sigma| > t \right\}.
\end{equation}
We use $\mathcal{E}_N$ to denote the event that some $\mathcal{E}_N (x,k)$ occurs, and bound the probability that it occurs as
\begin{align}
\P (\mathcal{E}_N) &\leq \sum_{x \in S,\, k \in [N]} \P (\mathcal{E}_N (x,k)) \nonumber\\
&\leq N 2^N \cdot \P \left( |D_{1}(\mathbf{1})/\sigma| > t \right) \nonumber\\ 
&\leq N 2^N \cdot 2 e^{-(\log 2)N - 4\log N} \nonumber\\ 
&= 2N \cdot N^{-4} \leq N^{-2}.
\end{align}
The first inequality is a union bound over $x$ and $k$. The second bound holds by the definition of the event $\mathcal{E}_N (x,k)$ and the fact that the $k$th energy difference of state $x$ has the same distribution as the first energy difference of state $\mathbf{1} = (1,\dots,1) \in \{-1,1\}^N$. The third inequality is due to the tail bound in \cref{eq:d tail bd}. The last inequality holds by the assumption that $N \geq 2$. Hence, except with a probability of $1-N^{-2}$, 
\begin{equation}
\max_{x \in S,\, k \in [N]} |D_k (x)| \leq \sigma t \leq 4 \sqrt{N}.
\end{equation}
\end{proof}

\subsection{Approximation of \texorpdfstring{$A$}{A} by \texorpdfstring{$\wt A$}{tilde A}}

Following the discussion at the beginning of \cref{app:sk}, we proceed to bound the variance of $A(X) - \wt A(X)$. We do so using the approximation in \cref{eq:summand_bound}, which is justified for small $\beta$ due to \cref{lem:e occurs}. We continue to assume that $N \geq 2$ is an integer.

\begin{lemma}\label{lem:vareps}
    As $\beta \to 0$, the quantity $\eps = A(X) - \wt A(X)$ satisfies
    \begin{equation}
    \frac{1}{N}\sqrt{\Var (\eps)} = O (\beta),
    \end{equation}
    with a probability of at least $1 - N^{-2}$.
\end{lemma}

\begin{proof}
    Suppose that the largest energy difference $D_k (x)$ is at most $4\sqrt{N}$ in absolute value. According to \cref{lem:e occurs}, this event occurs with a probability of at least $1 - N^{-2}$. When it occurs, we can take $\beta \to 0$ and expand $A$ as
    \begin{align}
    A(x) &= \frac{1}{\beta} \log \sum_{k \in [N]} \left(1 - \beta (1-\lambda) D_k (x) + N \cdot O(\beta^2) \right) \nonumber\\ 
    &= \frac{1}{\beta} \log N - \frac{1-\lambda}{N} \sum_{k \in [N]} D_k (x) + N \cdot O(\beta).
    \end{align}
    Hence, $\eps = N \cdot O(\beta)$ and therefore
    \begin{equation}
    \frac{1}{N} \sqrt{\Var(\eps)} \leq \frac{1}{N} \sqrt{\E(\eps^2)} = O(\beta).
    \end{equation}
\end{proof}

\subsection{Variance and covariance calculations}\label{app:deferred_moment_estimates}

The next inputs to the proof of \cref{thm:sk} are the following expressions for the variances and covariance of $H(X)$ and $\wt A(X)$. We will use these expressions both to calculate $\wt\rho$ and to show that $\wt\rho$ is close to $\rho$.

\begin{lemma}\label{lem:sk_moments}
    Given couplings $(g_{ij})_{i,j \in [N]}$, the variances and covariance of $H(X)$ and $\wt A(X)$ satisfy
    \begin{align}
        \Var(H(X)) &= \sum_{1 \leq i < j \leq N} (g_{ij} + g_{ji})^2,\\ 
        \Var(\wt A(X)) &= b^2 \Var (H(X)), \\ 
        \Cov(H(X),\wt A(X)) &= b \Var (H(X)),
    \end{align}
    where $b = 4(1-\lambda)/N$.
\end{lemma}

We note that since $H$ is a pseudo-Boolean function, i.e., a real-valued function on the hypercube $\{-1,1\}^N$, it is possible to calculate $\Var H(X)$ using Parseval's theorem \cite{ODonnell2014}. Instead, we calculate the moments directly, both for the sake of exposition and to facilitate the subsequent calculation of $\Var (\wt A(X))$.

\begin{proof}[Proof of \cref{lem:sk_moments}]
We start with the variance of $H(X) = \sum_{i,j \in [N]} g_{ij} X_i X_j$. Its first moment equals
\begin{equation}
    \E (H(X)) = \sum_{i,j \in [N]} g_{ij} \E (X_i X_j) = \sum_{i \in [N]} g_{ii},
\end{equation}
because $\E(X_i X_j)$ equals $1$ when $i = j$, and $0$ otherwise. The second moment satisfies
\begin{equation}
    \E \left( H(X)^2 \right) = \sum_{i,j,k,l \in [N]} g_{ij}g_{kl} \E (X_i X_j X_k X_l).
\end{equation}
The quantity $\E (X_i X_j X_k X_l)$ equals $0$ unless $\{i,j,k,l\}$ contains one or two distinct coordinates, in which case it equals $1$. By considering the various cases, we find that
\begin{equation}
    \E \left( H(X)^2 \right) = \sum_{i \in [N]} g_{ii}^2 + \sum_{i \neq j} (g_{ii} g_{jj} + g_{ij}^2 + g_{ij} g_{ji}).
\end{equation}
Hence, the variance satisfies
\begin{align}
    \Var (H(X)) &= \sum_{i \neq j} (g_{ij}^2 + g_{ij} g_{ji}) \nonumber \\ 
    &= \sum_{1 \leq i < j \leq N} (g_{ij} + g_{ji})^2.
\end{align}

Next, we calculate $\Var (\wt A (X))$, as follows:
\begin{align}
\Var (\wt A(X)) &= \Var \Bigg( \frac{1}{\beta} \log N - \frac{1-\lambda}{N} \sum_{k \in [N]} D_k (X) \Bigg) \nonumber\\ 
&= \frac{(1-\lambda)^2}{N^2} \Var \Bigg( \sum_{k \in [N]} D_k (X) \Bigg) \nonumber\\ 
&= \frac{4 (1-\lambda)^2}{N^2} \Var \Bigg( \sum_{i \neq j} (g_{ij} + g_{ji}) X_i X_j \Bigg) \nonumber\\ 
&= \frac{16 (1-\lambda)^2}{N^2} \Var \Bigg( \sum_{i \neq j} g_{ij} X_i X_j \Bigg) \nonumber\\ 
&= b^2 \Var (H(X)).
\end{align}
The first equality is due to the definition of $\wt A(X)$. The second holds because the constant $(\log N)/\beta$ does not affect the variance. The third follows from the formula for $D_k (x)$ in \cref{eq:diff_formula}. The fourth holds because the sums of $g_{ij} X_i X_j$ and $g_{ji} X_i X_j$ are equal. The final equality is due to the fact that the quantities $\sum_{i \neq j} g_{ij} X_i X_j$ and $H(X)$ differ only by the constant $\sum_{i \in [N]} g_{ii}$, which does not affect the variance.

By simply repeating these steps for the covariance, we find that 
\begin{equation}
\Cov (H(X), \wt A(X)) = b \Var (H(X)),
\end{equation}
which concludes the proof.
\end{proof}

Note that, as a consequence of \cref{lem:sk_moments}, we know the distribution of $\Var (H(X))$ exactly.

\begin{lemma}\label{lem:dist_of_varh}
    Let $N \geq 2$ be an integer. The distribution of $\Var(H(X))$ is $\mathrm{Gamma}(\frac{N(N-1)}{4},\frac{4}{N})$.
\end{lemma}

\begin{proof}
    Since the couplings are i.i.d.\ $\mathcal{N}(0,\frac{1}{N})$ random variables, the quantity $(g_{ij} + g_{ji})^2$ has a $\mathrm{Gamma}(\frac{1}{2},\frac{4}{N})$ distribution. By \cref{lem:sk_moments}, $\Var (H(X))$ is a sum of $\frac{N(N-1)}{2}$ independent such random variables, hence it has a $\mathrm{Gamma}(\frac{N(N-1)}{4},\frac{4}{N})$ distribution.
\end{proof}

A consequence of \cref{lem:dist_of_varh} is that $\Var(H(X))$ has a mean of $N-1$ and a variance of $4(1-\frac{1}{N})$ with respect to the randomness of the couplings, so it rarely takes values less than $(N-1)/2$. Specifically, the standard Chernoff bound for gamma random variables implies that
\begin{equation}\label{eq:chernoff bd}
\P \left(\Var (H(X)) < \frac{N-1}{2} \right) \leq e^{-c N (N-1)},
\end{equation}
where $c = (\log 2 - 1/2)/4 \approx 0.048$.

\subsection{Conclusion}

We now combine the results from the two preceding subsections to prove \cref{thm:sk}.

\begin{proof}[Proof of \cref{thm:sk}]

For the moment, assume that the couplings are given. We write $\rho$ as $\Corr (U,V)$, where
\begin{equation}
U = H(X) \quad \text{and} \quad V = (2\lambda-1)H(X) + A(X).
\end{equation}
Analogously, we write $\wt\rho = \Corr (U,\wt V)$, in terms of
\begin{equation}
\wt V = (2\lambda-1) H(X) + \wt A(X).
\end{equation}
By \cref{lem:sk_moments}, the corresponding covariance satisfies
\begin{align}
\Cov (U,\wt V) &= (2\lambda-1) \Var (H(X)) + \Cov(H(X),\wt A(X)) \nonumber\\ 
&= \left( (2\lambda - 1) + b \right) \Var (H(X)), \label{eq: sk cov uv tilde}
\end{align}
in terms of $b = 4(1-\lambda)/N$. Similarly, the variance of $\wt V$ satisfies
\begin{multline}
\Var (\wt V) = (2\lambda-1)^2 \Var (H(X)) + \Var (\wt A(X))\\ + 2 (2\lambda-1) \Cov (H(X),\wt A(X)),
\end{multline}
so \cref{lem:sk_moments} implies that
\begin{align}
\Var (\wt V) &= \left( (2\lambda - 1)^2 + 2(2\lambda -1 ) b + b^2 \right) \Var (H(X)) \nonumber\\ 
&= ((2\lambda - 1) + b)^2 \Var (H(X)).  \label{eq: sk var v tilde}
\end{align}
Let $c = |2\lambda - 1 + b| \sqrt{\Var (H(X))}$. \cref{eq: sk cov uv tilde} shows that 
\begin{equation}\label{eq:verif hypos1}
| \Cov (U, \wt V) | \geq c \sqrt{\Var (U)},
\end{equation}
while \cref{eq: sk var v tilde} establishes that
\begin{equation}\label{eq:verif hypos2}
\Var(\wt V) \geq c^2.
\end{equation}
Note that
\begin{equation}
2\lambda - 1 + b = 2\left( 1 - \frac{2}{N} \right) (\lambda-\lambda_\ast),
\end{equation}
so $c$ depends on the absolute distance of $\lambda$ from the location of the threshold $\lambda_\ast$ in \cref{thm:sk}.

Next, to verify the hypotheses of \cref{lem:corr_perturb}, we address the randomness of the couplings. \cref{lem:vareps} and \cref{eq:chernoff bd} together imply that the event consisting of
\begin{align}
\frac{1}{N}\sqrt{\Var(\eps)} &= O(\beta) \,\, \text{as $\beta \to 0$}, \nonumber \\ 
\Var (H(X)) &\geq \frac{N-1}{2}
\end{align}
occurs with a probability of at least $1 - O(N^{-2})$ as $N~\to~\infty$. When it does, $c$ is positive for all $\lambda \neq \lambda_\ast$ because 
\begin{align}
c &\geq 2 \left(1 - \frac{2}{N} \right) |\lambda-\lambda_\ast| \sqrt{\frac{N-1}{2}} \nonumber\\ 
&\geq \frac{2}{3\sqrt{3}} |\lambda - \lambda_\ast| \sqrt{N}.
\end{align}
Note that the second inequality follows from the simple bounds $(1-2/N) \geq 1/3$ and $\sqrt{(N-1)/2} \geq \sqrt{N/3}$, which hold for all $N \geq 3$. \cref{eq:verif hypos1,eq:verif hypos2} show that $U$ and $\wt V$ satisfy the hypotheses of \cref{lem:corr_perturb} for this choice of $c$, for all $\lambda \neq \lambda_\ast$. \cref{lem:corr_perturb} then implies that, as $\beta \to 0$, $\rho$ equals
\begin{equation}
\wt \rho \left(1 \pm \frac{5}{c} \sqrt{\Var (\eps)} \right) = \wt\rho \left(1 + \frac{\sqrt{N}}{|\lambda-\lambda_\ast|} \cdot O(\beta) \right),
\end{equation}
for all such $\lambda$. This completes the proof, since $\wt\rho = \mathrm{sign}(\lambda-\lambda_\ast)$ for all $\lambda \neq \lambda_\ast$ by \cref{eq:sk rho tilde}.
\end{proof}

\bibliographystyle{apsrev4-2}

\begin{thebibliography}{25}%
\makeatletter
\providecommand \@ifxundefined [1]{%
 \@ifx{#1\undefined}
}%
\providecommand \@ifnum [1]{%
 \ifnum #1\expandafter \@firstoftwo
 \else \expandafter \@secondoftwo
 \fi
}%
\providecommand \@ifx [1]{%
 \ifx #1\expandafter \@firstoftwo
 \else \expandafter \@secondoftwo
 \fi
}%
\providecommand \natexlab [1]{#1}%
\providecommand \enquote  [1]{``#1''}%
\providecommand \bibnamefont  [1]{#1}%
\providecommand \bibfnamefont [1]{#1}%
\providecommand \citenamefont [1]{#1}%
\providecommand \href@noop [0]{\@secondoftwo}%
\providecommand \href [0]{\begingroup \@sanitize@url \@href}%
\providecommand \@href[1]{\@@startlink{#1}\@@href}%
\providecommand \@@href[1]{\endgroup#1\@@endlink}%
\providecommand \@sanitize@url [0]{\catcode `\\12\catcode `\$12\catcode `\&12\catcode `\#12\catcode `\^12\catcode `\_12\catcode `\%12\relax}%
\providecommand \@@startlink[1]{}%
\providecommand \@@endlink[0]{}%
\providecommand \url  [0]{\begingroup\@sanitize@url \@url }%
\providecommand \@url [1]{\endgroup\@href {#1}{\urlprefix }}%
\providecommand \urlprefix  [0]{URL }%
\providecommand \Eprint [0]{\href }%
\providecommand \doibase [0]{https://doi.org/}%
\providecommand \selectlanguage [0]{\@gobble}%
\providecommand \bibinfo  [0]{\@secondoftwo}%
\providecommand \bibfield  [0]{\@secondoftwo}%
\providecommand \translation [1]{[#1]}%
\providecommand \BibitemOpen [0]{}%
\providecommand \bibitemStop [0]{}%
\providecommand \bibitemNoStop [0]{.\EOS\space}%
\providecommand \EOS [0]{\spacefactor3000\relax}%
\providecommand \BibitemShut  [1]{\csname bibitem#1\endcsname}%
\let\auto@bib@innerbib\@empty
\bibitem [{\citenamefont {Dill}(1990)}]{Dill1990}%
  \BibitemOpen
  \bibfield  {author} {\bibinfo {author} {\bibfnamefont {K.~A.}\ \bibnamefont {Dill}},\ }\href {https://doi.org/10.1021/bi00483a001} {\bibfield  {journal} {\bibinfo  {journal} {Biochemistry}\ }\textbf {\bibinfo {volume} {29}},\ \bibinfo {pages} {7133} (\bibinfo {year} {1990})}\BibitemShut {NoStop}%
\bibitem [{\citenamefont {Chvykov}\ \emph {et~al.}(2021)\citenamefont {Chvykov}, \citenamefont {Berrueta}, \citenamefont {Vardhan}, \citenamefont {Savoie}, \citenamefont {Samland}, \citenamefont {Murphey}, \citenamefont {Wiesenfeld}, \citenamefont {Goldman},\ and\ \citenamefont {England}}]{Chvykov2021}%
  \BibitemOpen
  \bibfield  {author} {\bibinfo {author} {\bibfnamefont {P.}~\bibnamefont {Chvykov}}, \bibinfo {author} {\bibfnamefont {T.~A.}\ \bibnamefont {Berrueta}}, \bibinfo {author} {\bibfnamefont {A.}~\bibnamefont {Vardhan}}, \bibinfo {author} {\bibfnamefont {W.}~\bibnamefont {Savoie}}, \bibinfo {author} {\bibfnamefont {A.}~\bibnamefont {Samland}}, \bibinfo {author} {\bibfnamefont {T.~D.}\ \bibnamefont {Murphey}}, \bibinfo {author} {\bibfnamefont {K.}~\bibnamefont {Wiesenfeld}}, \bibinfo {author} {\bibfnamefont {D.~I.}\ \bibnamefont {Goldman}},\ and\ \bibinfo {author} {\bibfnamefont {J.~L.}\ \bibnamefont {England}},\ }\href {https://doi.org/10.1126/science.abc6182} {\bibfield  {journal} {\bibinfo  {journal} {Science}\ }\textbf {\bibinfo {volume} {371}},\ \bibinfo {pages} {90} (\bibinfo {year} {2021})}\BibitemShut {NoStop}%
\bibitem [{\citenamefont {Noé}\ \emph {et~al.}(2019)\citenamefont {Noé}, \citenamefont {Olsson}, \citenamefont {Köhler},\ and\ \citenamefont {Wu}}]{noe_boltzmann_2019}%
  \BibitemOpen
  \bibfield  {author} {\bibinfo {author} {\bibfnamefont {F.}~\bibnamefont {Noé}}, \bibinfo {author} {\bibfnamefont {S.}~\bibnamefont {Olsson}}, \bibinfo {author} {\bibfnamefont {J.}~\bibnamefont {Köhler}},\ and\ \bibinfo {author} {\bibfnamefont {H.}~\bibnamefont {Wu}},\ }\href {https://doi.org/10.1126/science.aaw1147} {\bibfield  {journal} {\bibinfo  {journal} {Science}\ }\textbf {\bibinfo {volume} {365}},\ \bibinfo {pages} {eaaw1147} (\bibinfo {year} {2019})}\BibitemShut {NoStop}%
\bibitem [{\citenamefont {Lee}\ \emph {et~al.}(2013)\citenamefont {Lee}, \citenamefont {Ozdaglar},\ and\ \citenamefont {Shah}}]{lee_computing_2013}%
  \BibitemOpen
  \bibfield  {author} {\bibinfo {author} {\bibfnamefont {C.~E.}\ \bibnamefont {Lee}}, \bibinfo {author} {\bibfnamefont {A.}~\bibnamefont {Ozdaglar}},\ and\ \bibinfo {author} {\bibfnamefont {D.}~\bibnamefont {Shah}},\ }in\ \href {https://proceedings.neurips.cc/paper_files/paper/2013/hash/99bcfcd754a98ce89cb86f73acc04645-Abstract.html} {\emph {\bibinfo {booktitle} {Advances in {Neural} {Information} {Processing} {Systems}}}},\ Vol.~\bibinfo {volume} {26}\ (\bibinfo  {publisher} {Curran Associates, Inc.},\ \bibinfo {year} {2013})\BibitemShut {NoStop}%
\bibitem [{\citenamefont {Landauer}(1975)}]{landauer_inadequacy_1975}%
  \BibitemOpen
  \bibfield  {author} {\bibinfo {author} {\bibfnamefont {R.}~\bibnamefont {Landauer}},\ }\href {https://doi.org/10.1103/PhysRevA.12.636} {\bibfield  {journal} {\bibinfo  {journal} {Physical Review A}\ }\textbf {\bibinfo {volume} {12}},\ \bibinfo {pages} {636} (\bibinfo {year} {1975})}\BibitemShut {NoStop}%
\bibitem [{\citenamefont {Landauer}(1988)}]{landauer_motion_1988}%
  \BibitemOpen
  \bibfield  {author} {\bibinfo {author} {\bibfnamefont {R.}~\bibnamefont {Landauer}},\ }\href {https://doi.org/10.1007/BF01011555} {\bibfield  {journal} {\bibinfo  {journal} {Journal of Statistical Physics}\ }\textbf {\bibinfo {volume} {53}},\ \bibinfo {pages} {233} (\bibinfo {year} {1988})}\BibitemShut {NoStop}%
\bibitem [{\citenamefont {Chvykov}\ and\ \citenamefont {England}(2018)}]{Chvykov2018}%
  \BibitemOpen
  \bibfield  {author} {\bibinfo {author} {\bibfnamefont {P.}~\bibnamefont {Chvykov}}\ and\ \bibinfo {author} {\bibfnamefont {J.}~\bibnamefont {England}},\ }\href {https://doi.org/10.1103/PhysRevE.97.032115} {\bibfield  {journal} {\bibinfo  {journal} {Physical Review E}\ }\textbf {\bibinfo {volume} {97}},\ \bibinfo {pages} {032115} (\bibinfo {year} {2018})}\BibitemShut {NoStop}%
\bibitem [{\citenamefont {Gold}(2021)}]{gold_self-organized_2021}%
  \BibitemOpen
  \bibfield  {author} {\bibinfo {author} {\bibfnamefont {J.}~\bibnamefont {Gold}},\ }\emph {\bibinfo {title} {Self-organized fine-tuned response in a driven spin glass}},\ \href@noop {} {\bibinfo {type} {Ph.{D}. thesis}},\ \bibinfo  {school} {Massachusetts Institute of Technology} (\bibinfo {year} {2021})\BibitemShut {NoStop}%
\bibitem [{\citenamefont {Yang}\ \emph {et~al.}(2022)\citenamefont {Yang}, \citenamefont {Berrueta}, \citenamefont {Brooks}, \citenamefont {Liu}, \citenamefont {Zhang}, \citenamefont {Gonzalez-Medrano}, \citenamefont {Yang}, \citenamefont {Koman}, \citenamefont {Chvykov}, \citenamefont {LeMar}, \citenamefont {Miskin}, \citenamefont {Murphey},\ and\ \citenamefont {Strano}}]{yang_emergent_2022}%
  \BibitemOpen
  \bibfield  {author} {\bibinfo {author} {\bibfnamefont {J.~F.}\ \bibnamefont {Yang}}, \bibinfo {author} {\bibfnamefont {T.~A.}\ \bibnamefont {Berrueta}}, \bibinfo {author} {\bibfnamefont {A.~M.}\ \bibnamefont {Brooks}}, \bibinfo {author} {\bibfnamefont {A.~T.}\ \bibnamefont {Liu}}, \bibinfo {author} {\bibfnamefont {G.}~\bibnamefont {Zhang}}, \bibinfo {author} {\bibfnamefont {D.}~\bibnamefont {Gonzalez-Medrano}}, \bibinfo {author} {\bibfnamefont {S.}~\bibnamefont {Yang}}, \bibinfo {author} {\bibfnamefont {V.~B.}\ \bibnamefont {Koman}}, \bibinfo {author} {\bibfnamefont {P.}~\bibnamefont {Chvykov}}, \bibinfo {author} {\bibfnamefont {L.~N.}\ \bibnamefont {LeMar}}, \bibinfo {author} {\bibfnamefont {M.~Z.}\ \bibnamefont {Miskin}}, \bibinfo {author} {\bibfnamefont {T.~D.}\ \bibnamefont {Murphey}},\ and\ \bibinfo {author} {\bibfnamefont {M.~S.}\ \bibnamefont {Strano}},\ }\href {https://doi.org/10.1038/s41467-022-33396-5} {\bibfield  {journal} {\bibinfo  {journal} {Nature Communications}\ }\textbf {\bibinfo {volume}
  {13}},\ \bibinfo {pages} {5734} (\bibinfo {year} {2022})}\BibitemShut {NoStop}%
\bibitem [{\citenamefont {Kedia}\ \emph {et~al.}(2023)\citenamefont {Kedia}, \citenamefont {Pan}, \citenamefont {Slotine},\ and\ \citenamefont {England}}]{kedia_drive-specific_2023}%
  \BibitemOpen
  \bibfield  {author} {\bibinfo {author} {\bibfnamefont {H.}~\bibnamefont {Kedia}}, \bibinfo {author} {\bibfnamefont {D.}~\bibnamefont {Pan}}, \bibinfo {author} {\bibfnamefont {J.-J.}\ \bibnamefont {Slotine}},\ and\ \bibinfo {author} {\bibfnamefont {J.~L.}\ \bibnamefont {England}},\ }\href {https://doi.org/10.1063/5.0171993} {\bibfield  {journal} {\bibinfo  {journal} {The Journal of Chemical Physics}\ }\textbf {\bibinfo {volume} {159}},\ \bibinfo {pages} {214106} (\bibinfo {year} {2023})}\BibitemShut {NoStop}%
\bibitem [{\citenamefont {Calvert}\ and\ \citenamefont {Randall}(2024)}]{Calvert2024}%
  \BibitemOpen
  \bibfield  {author} {\bibinfo {author} {\bibfnamefont {J.}~\bibnamefont {Calvert}}\ and\ \bibinfo {author} {\bibfnamefont {D.}~\bibnamefont {Randall}},\ }\href {https://doi.org/10.1073/pnas.2411731121} {\bibfield  {journal} {\bibinfo  {journal} {Proceedings of the National Academy of Sciences}\ }\textbf {\bibinfo {volume} {121}},\ \bibinfo {pages} {e2411731121} (\bibinfo {year} {2024})}\BibitemShut {NoStop}%
\bibitem [{\citenamefont {Schnakenberg}(1976)}]{schnakenberg_network_1976}%
  \BibitemOpen
  \bibfield  {author} {\bibinfo {author} {\bibfnamefont {J.}~\bibnamefont {Schnakenberg}},\ }\href {https://doi.org/10.1103/RevModPhys.48.571} {\bibfield  {journal} {\bibinfo  {journal} {Reviews of Modern Physics}\ }\textbf {\bibinfo {volume} {48}},\ \bibinfo {pages} {571} (\bibinfo {year} {1976})}\BibitemShut {NoStop}%
\bibitem [{\citenamefont {Zia}\ and\ \citenamefont {Schmittmann}(2007)}]{zia_probability_2007}%
  \BibitemOpen
  \bibfield  {author} {\bibinfo {author} {\bibfnamefont {R.~K.~P.}\ \bibnamefont {Zia}}\ and\ \bibinfo {author} {\bibfnamefont {B.}~\bibnamefont {Schmittmann}},\ }\href {https://doi.org/10.1088/1742-5468/2007/07/P07012} {\bibfield  {journal} {\bibinfo  {journal} {Journal of Statistical Mechanics: Theory and Experiment}\ }\textbf {\bibinfo {volume} {2007}},\ \bibinfo {pages} {P07012} (\bibinfo {year} {2007})}\BibitemShut {NoStop}%
\bibitem [{\citenamefont {Barato}\ and\ \citenamefont {Seifert}(2015)}]{Barato2015}%
  \BibitemOpen
  \bibfield  {author} {\bibinfo {author} {\bibfnamefont {A.~C.}\ \bibnamefont {Barato}}\ and\ \bibinfo {author} {\bibfnamefont {U.}~\bibnamefont {Seifert}},\ }\href {https://doi.org/10.1103/PhysRevLett.114.158101} {\bibfield  {journal} {\bibinfo  {journal} {Physical Review Letters}\ }\textbf {\bibinfo {volume} {114}},\ \bibinfo {pages} {158101} (\bibinfo {year} {2015})}\BibitemShut {NoStop}%
\bibitem [{\citenamefont {Owen}\ \emph {et~al.}(2020)\citenamefont {Owen}, \citenamefont {Gingrich},\ and\ \citenamefont {Horowitz}}]{owen_universal_2020}%
  \BibitemOpen
  \bibfield  {author} {\bibinfo {author} {\bibfnamefont {J.~A.}\ \bibnamefont {Owen}}, \bibinfo {author} {\bibfnamefont {T.~R.}\ \bibnamefont {Gingrich}},\ and\ \bibinfo {author} {\bibfnamefont {J.~M.}\ \bibnamefont {Horowitz}},\ }\href {https://doi.org/10.1103/PhysRevX.10.011066} {\bibfield  {journal} {\bibinfo  {journal} {Physical Review X}\ }\textbf {\bibinfo {volume} {10}},\ \bibinfo {pages} {011066} (\bibinfo {year} {2020})}\BibitemShut {NoStop}%
\bibitem [{\citenamefont {Falasco}\ and\ \citenamefont {Esposito}(2025)}]{RevModPhys.97.015002}%
  \BibitemOpen
  \bibfield  {author} {\bibinfo {author} {\bibfnamefont {G.}~\bibnamefont {Falasco}}\ and\ \bibinfo {author} {\bibfnamefont {M.}~\bibnamefont {Esposito}},\ }\href {https://doi.org/10.1103/RevModPhys.97.015002} {\bibfield  {journal} {\bibinfo  {journal} {Reviews of Modern Physics}\ }\textbf {\bibinfo {volume} {97}},\ \bibinfo {pages} {015002} (\bibinfo {year} {2025})}\BibitemShut {NoStop}%
\bibitem [{\citenamefont {Calvert}\ and\ \citenamefont {Randall}(2025)}]{calvert2025localglobalcorrelationsdynamicsdisordered}%
  \BibitemOpen
  \bibfield  {author} {\bibinfo {author} {\bibfnamefont {J.}~\bibnamefont {Calvert}}\ and\ \bibinfo {author} {\bibfnamefont {D.}~\bibnamefont {Randall}},\ }\href {https://arxiv.org/abs/2508.04501} {\bibinfo {title} {Local--global correlations of dynamics on disordered energy landscapes}} (\bibinfo {year} {2025}),\ \Eprint {https://arxiv.org/abs/2508.04501} {arXiv:2508.04501 [math.PR]} \BibitemShut {NoStop}%
\bibitem [{\citenamefont {Norris}(1997)}]{norris_markov_1997}%
  \BibitemOpen
  \bibfield  {author} {\bibinfo {author} {\bibfnamefont {J.~R.}\ \bibnamefont {Norris}},\ }\href {https://doi.org/10.1017/CBO9780511810633} {\emph {\bibinfo {title} {Markov {Chains}}}},\ Cambridge {Series} in {Statistical} and {Probabilistic} {Mathematics}\ (\bibinfo  {publisher} {Cambridge University Press},\ \bibinfo {address} {Cambridge},\ \bibinfo {year} {1997})\BibitemShut {NoStop}%
\bibitem [{\citenamefont {Mathieu}(2000)}]{mathieu_convergence_2000}%
  \BibitemOpen
  \bibfield  {author} {\bibinfo {author} {\bibfnamefont {P.}~\bibnamefont {Mathieu}},\ }\href {https://doi.org/10.1007/s002200000292} {\bibfield  {journal} {\bibinfo  {journal} {Communications in Mathematical Physics}\ }\textbf {\bibinfo {volume} {215}},\ \bibinfo {pages} {57} (\bibinfo {year} {2000})}\BibitemShut {NoStop}%
\bibitem [{\citenamefont {Calvert}\ \emph {et~al.}(2025)\citenamefont {Calvert}, \citenamefont {den Hollander},\ and\ \citenamefont {Randall}}]{calvert2025noteasymptoticuniformitymarkov}%
  \BibitemOpen
  \bibfield  {author} {\bibinfo {author} {\bibfnamefont {J.}~\bibnamefont {Calvert}}, \bibinfo {author} {\bibfnamefont {F.}~\bibnamefont {den Hollander}},\ and\ \bibinfo {author} {\bibfnamefont {D.}~\bibnamefont {Randall}},\ }\href {https://arxiv.org/abs/2505.01608} {\bibinfo {title} {A note on the asymptotic uniformity of {M}arkov chains with random rates}} (\bibinfo {year} {2025}),\ \Eprint {https://arxiv.org/abs/2505.01608} {arXiv:2505.01608 [math.PR]} \BibitemShut {NoStop}%
\bibitem [{\citenamefont {Leighton}\ and\ \citenamefont {Rivest}(1986)}]{leighton_estimating_1986}%
  \BibitemOpen
  \bibfield  {author} {\bibinfo {author} {\bibfnamefont {F.}~\bibnamefont {Leighton}}\ and\ \bibinfo {author} {\bibfnamefont {R.}~\bibnamefont {Rivest}},\ }\href {https://doi.org/10.1109/TIT.1986.1057250} {\bibfield  {journal} {\bibinfo  {journal} {IEEE Transactions on Information Theory}\ }\textbf {\bibinfo {volume} {32}},\ \bibinfo {pages} {733} (\bibinfo {year} {1986})}\BibitemShut {NoStop}%
\bibitem [{\citenamefont {Dal~Cengio}\ \emph {et~al.}(2023)\citenamefont {Dal~Cengio}, \citenamefont {Lecomte},\ and\ \citenamefont {Polettini}}]{DalCengio2023}%
  \BibitemOpen
  \bibfield  {author} {\bibinfo {author} {\bibfnamefont {S.}~\bibnamefont {Dal~Cengio}}, \bibinfo {author} {\bibfnamefont {V.}~\bibnamefont {Lecomte}},\ and\ \bibinfo {author} {\bibfnamefont {M.}~\bibnamefont {Polettini}},\ }\href {https://doi.org/10.1103/PhysRevX.13.021040} {\bibfield  {journal} {\bibinfo  {journal} {Physical Review X}\ }\textbf {\bibinfo {volume} {13}},\ \bibinfo {pages} {021040} (\bibinfo {year} {2023})}\BibitemShut {NoStop}%
\bibitem [{\citenamefont {Floyd}\ \emph {et~al.}(2025)\citenamefont {Floyd}, \citenamefont {Dinner}, \citenamefont {Murugan},\ and\ \citenamefont {Vaikuntanathan}}]{Floyd2025}%
  \BibitemOpen
  \bibfield  {author} {\bibinfo {author} {\bibfnamefont {C.}~\bibnamefont {Floyd}}, \bibinfo {author} {\bibfnamefont {A.~R.}\ \bibnamefont {Dinner}}, \bibinfo {author} {\bibfnamefont {A.}~\bibnamefont {Murugan}},\ and\ \bibinfo {author} {\bibfnamefont {S.}~\bibnamefont {Vaikuntanathan}},\ }\href {https://doi.org/10.1038/s41467-025-61873-0} {\bibfield  {journal} {\bibinfo  {journal} {Nature Communications}\ }\textbf {\bibinfo {volume} {16}},\ \bibinfo {pages} {7184} (\bibinfo {year} {2025})}\BibitemShut {NoStop}%
\bibitem [{\citenamefont {Calvert}(2025)}]{calvert_2025_17791154}%
  \BibitemOpen
  \bibfield  {author} {\bibinfo {author} {\bibfnamefont {J.}~\bibnamefont {Calvert}},\ }\href@noop {} {\bibinfo {title} {Local--global correlation thresholds (v1.0)}},\ \bibinfo {howpublished} {\url{https://doi.org/10.5281/zenodo.17791154}} (\bibinfo {year} {2025})\BibitemShut {NoStop}%
\bibitem [{\citenamefont {O'Donnell}(2014)}]{ODonnell2014}%
  \BibitemOpen
  \bibfield  {author} {\bibinfo {author} {\bibfnamefont {R.}~\bibnamefont {O'Donnell}},\ }\href {https://doi.org/DOI: 10.1017/CBO9781139814782} {\emph {\bibinfo {title} {Analysis of Boolean Functions}}}\ (\bibinfo  {publisher} {Cambridge University Press},\ \bibinfo {address} {Cambridge},\ \bibinfo {year} {2014})\BibitemShut {NoStop}%
\end{thebibliography}

%

\end{document}